\documentclass[11pt,a4paper]{article}

\usepackage[utf8]{inputenc}
\usepackage[T1]{fontenc}
\usepackage{lmodern}
\usepackage{amsmath,amssymb,amsthm,mathtools}
\usepackage{xcolor}
\usepackage{enumitem}
\usepackage[round,authoryear]{natbib}
\usepackage{geometry}
\usepackage{hyperref}
\usepackage{microtype}

\hypersetup{
  colorlinks=true,
  linkcolor=blue!55!black,
  citecolor=blue!55!black,
  urlcolor=blue!55!black,
  pdftitle={A 1.283 price-of-anarchy bound for the repeated virtual first-price auction},
  pdfauthor={Endre Cs{\'o}ka},
  pdfsubject={Repeated first-price allocation with artificial credits, robust utility guarantees, and a 1.283 price-of-anarchy bound},
  pdfkeywords={repeated allocation, artificial currency, first-price auction, price of anarchy, maximin utility, collusion-robustness}
}

\newcommand{\Bidder}{\textnormal{\textsc{Bidder}}}
\newcommand{\Adversary}{\textnormal{\textsc{Adversary}}}
\newcommand{\BidderTag}{\mathrm{B}}
\newcommand{\AdversaryTag}{\mathrm{A}}
\newcommand{\E}{\mathbb{E}}
\newcommand{\one}{\mathbf{1}}
\newcommand{\lowC}{\underline C}
\newcommand{\upC}{\overline C}

\theoremstyle{plain}
\newtheorem{theorem}{Theorem}
\newtheorem{lemma}{Lemma}
\newtheorem{proposition}{Proposition}
\newtheorem{corollary}{Corollary}

\theoremstyle{definition}
\newtheorem{definition}{Definition}
\newtheorem{assumption}{Assumption}

\theoremstyle{remark}
\newtheorem{remark}{Remark}

\title{A 1.283 Price-of-Anarchy Bound for the Repeated Virtual First-Price Auction}
\author{Endre Csóka\thanks{Supported by the NRDI grant KKP 138270.} \\ Alfréd Rényi Institute of Mathematics}

\begin{document}
\maketitle

\begin{abstract}
We study the repeated allocation of a single indivisible resource among $n$ strategic players.
Each player $i$ has a privately known value distribution $D_i$, and values are drawn independently across players and periods.
The goal is to find fair and efficient mechanisms.

We apply the repeated first-price auction with equal initial endowments of virtual money.
We show that each player can asymptotically secure the same fair-floor guarantee $f(D_i)$ as in \cite{Csoka2026}; consequently, the mechanism is $1.283$-optimal.
This provides a simpler and more robust alternative mechanism for this special case and may also help derive sharper upper bounds on the price of anarchy.
\end{abstract}



\section{Introduction}\label{sec:introduction}

\subsection{The mechanism and the headline result}

Many organizations repeatedly allocate a scarce object or service slot among users whose private urgency changes over time. Examples include computing capacity, laboratory equipment, internal service capacity, and other resources for which a cash market may be administratively undesirable or normatively inappropriate \citep{GorokhBanerjeeIyer2016}. A fixed rotation protects access but ignores how much users value different dates. Artificial currency offers a nonmonetary alternative: each participant receives credits, spends more when the current opportunity matters more, and bears the opportunity cost through reduced future purchasing power.

The mechanism studied here is the simplest version of that idea. Each player starts with a virtual budget. In every round the players submit feasible scalar bids, the highest bidder receives the object, and the winner's own bid is removed from her budget. Credits have no terminal value and never become real money; they only record how much future priority the winner is willing to sacrifice for the current allocation.

The headline result is an asymptotic price-of-anarchy bound of $1.283$. Because utility units are player-specific, the relevant comparison is multiplicative and Pareto rather than a sum of utilities. Given a normalized equilibrium utility vector $x$, say that its Pareto price of anarchy is at most $\lambda$ if there is no feasible expected-utility vector $u$ satisfying
\[
u_i>\lambda x_i
\qquad\text{for every player }i.
\]
We prove that every accumulation point of normalized expected-utility vectors of Bayesian Nash equilibria of the $T$-round auction has Pareto price of anarchy at most $1.283$. The same conclusion holds for $o(T)$-approximate equilibria.

This equilibrium statement follows from a stronger guarantee that does not assume equilibrium behavior. Let
\[
c_i^T
=
\sup_{\sigma_i}\inf_{\bar\sigma_{-i}}
\E\bigl[U_i^T(\sigma_i,\bar\sigma_{-i})\bigr]
\]
be player $i$'s security level, where the other players may use an arbitrary joint strategy, share information, and correlate their actions. We prove
\[
c_i^T\ge T\phi(\alpha_i,D_i)-o_{D_i}(T).
\]
At any $\delta$-Bayesian Nash equilibrium, player $i$ can deviate to a guaranteeing strategy, so her equilibrium utility is at least $c_i^T-\delta$. Thus every equilibrium inherits the playerwise guarantees.

The constant $1.283$ comes from a separate feasible-region theorem: no quota-free allocation rule can improve every coordinate of the fair-floor vector $(\phi(\alpha_i,D_i))_i$ by a factor strictly larger than $1.283$ \citep[Theorem~6.2]{Csoka2026}. The same work proves the complementary cross-instance statement that $T\phi$ is Pareto-optimal among universally feasible target functions \citep[Proposition~4.8]{Csoka2026}. The price-of-anarchy bound is therefore only the headline consequence. More fundamentally, the auction implements a Pareto-maximal system of robust playerwise guarantees through a fixed and familiar game form.

The information assumptions are essentially the same as in the prior-free construction: a player uses her own distribution $D_i$ to compute a guaranteeing strategy. The practical difference is in the mechanism. Here the distribution is not reported, the designer applies the same scalar first-price rule under every profile, and $D_i$ enters only the player's private bidding policy.

\subsection{The fair floor}

For a distribution $D$ with generalized quantile $Q_D$, the per-round fair floor is
\begin{equation}\label{eq:intro-fair-floor}
\phi(\alpha,D)=\int_0^1 Q_D(u)u^{1/\alpha-1}\,du.
\end{equation}
The exponent arises from a static weighted-rank allocation. Couple each player's value to an independent uniform rank $R_i$ and allocate the object to the largest score $R_i^{1/\alpha_i}$. Conditional on player $i$ having rank $r$, all other scores are smaller with probability
\[
\prod_{j\ne i}r^{\alpha_j/\alpha_i}=r^{1/\alpha_i-1}.
\]
Integrating $Q_{D_i}(r)$ against this winning probability gives $\phi(\alpha_i,D_i)$; integrating the probability alone gives $\alpha_i$. Thus the rule respects the budget shares as expected allocation frequencies while concentrating each player's scarce winning opportunities on her high-rank realizations.

Two formulas make the benchmark transparent. If $\alpha=1/n$, then
\[
\phi(1/n,D)=\frac1n\E\left[\max_{1\le k\le n}V_k\right]
\]
for independent $V_k\sim D$: the floor is the equal share of efficient welfare among $n$ symmetric replicas. If $U_\alpha\sim\operatorname{Beta}(1/\alpha,1)$, then
\[
\phi(\alpha,D)=\alpha\,\E[Q_D(U_\alpha)].
\]
A small share therefore gives a small entitlement concentrated on unusually high ranks rather than the unconditional proportional value $\alpha\E[V]$. The benchmark is \emph{fair} because it comes from a jointly feasible share-weighted allocation and divides efficient welfare equally in the symmetric replica case. The term \emph{floor} refers to its unilateral guarantee.

Two results locate this benchmark within the feasible utility region. First, $T\phi$ is Pareto-optimal among universally feasible target functions: no target can weakly improve it for every pair $(\alpha,D)$ and strictly improve it somewhere \citep[Proposition~4.8]{Csoka2026}. Second, for every fixed quota-free instance, no feasible expected-utility vector strictly dominates $1.283$ times the fair-floor vector coordinatewise \citep[Theorem~6.2]{Csoka2026}. The first statement gives the cross-instance maximality of the guarantees; the second yields the price-of-anarchy bound.

\subsection{The symmetric benchmark: why first price should work}\label{subsec:symmetric-heuristic}

The symmetric game gives the clearest explanation for why the theorem should be true. The heuristic has two parts. First, a common bid rule that is increasing in value rank should allocate the object almost efficiently. Second, support indifference in the maximin problem should make each player's utility under symmetric play close to her guaranteed utility.

A closely related monotonicity-and-support-indifference argument appears in the first-price auction example of \citet[Section~4.5]{CsokaPongraczRodivilov2025}. In that example bids are real transfers shared among the other buyers; here bids are burned virtual credits, and their force comes from the resulting change in future purchasing power.

Suppose
\[
\alpha_1=\cdots=\alpha_n=\frac1n,
\qquad
D_1=\cdots=D_n=D.
\]
For this heuristic, call a fixed player the \Bidder{} and aggregate all other players into an auxiliary \Adversary{} that controls their combined budget and seeks only to minimize the \Bidder{}'s utility. Let $C_T^{n,D}$ denote the \Bidder{}'s maximin value in this two-player game. Suppose every original player uses the same canonical maximin, or asymptotically near-maximin, policy, and let $U_T^{n,D}$ be one player's expected utility in this symmetric profile. Since each player is using a guaranteeing strategy,
\[
U_T^{n,D}\ge C_T^{n,D}-o(T).
\]

The canonical first-price bid is monotone in the player's randomized value rank. When all players are in the same state, a player with maximal value therefore submits a maximal bid, apart from atom and tie bookkeeping. The symmetric profile should consequently be nearly efficient:
\[
nU_T^{n,D}
=
T\E\left[\max_{1\le j\le n}V_j\right]-o(T).
\]

For the \Bidder{}, only the maximum of the other players' bids matters. Under symmetric play this maximum lies in the same support as the optimal mixed bid of the \Adversary{} in the two-player problem. At the saddle point, bids in that support leave the \Bidder{} indifferent. This suggests that symmetric play does not give the player substantially more than her security level:
\[
U_T^{n,D}=C_T^{n,D}+o(T).
\]
Combining efficiency and tightness gives
\[
C_T^{n,D}
=
\frac{T}{n}\E\left[\max_jV_j\right]-o(T)
=
T\phi(1/n,D)-o(T).
\]

The same picture also explains why first price is special. If the \Bidder{} loses a first-price round, the highest opposing bidder wins and pays exactly the highest opposing bid. Hence all opponents can be merged into the single \Adversary{} without changing either the winner or their total remaining budget. Under a second-price rule with at least two opponents, write $b_i$ for the \Bidder{}'s bid and let $b_{-i}^{(2)}$ be the second-highest bid among the opponents. Whenever at least two opponents outbid the \Bidder{}, the winner pays
\[
\max\{b_i,b_{-i}^{(2)}\}>b_i.
\]
The opponents then burn extra common budget through competition among themselves, whereas the \Adversary{} would need to pay only $b_i$ to defeat the \Bidder{}. Exact aggregation, and with it the tightness intuition above, is lost.

\subsection{Why the specialized implementation is useful}

The first benefit is practical. The mechanism asks only for one feasible bid per round and applies the same familiar rule under every distribution profile. Priors remain private inputs to decentralized bidding policies rather than messages that a central designer must collect and process.

Second, related artificial-credit systems have long been studied for allocation without money and have been implemented in settings such as course allocation \citep{BudishEtAl2017,GorokhBanerjeeIyer2021}. The theorem can therefore be read not only as a mechanism proposal, but also as a guarantee-theoretic justification for a recognizable practice: give users credits and let them spend those credits in repeated priority contests. Within the model, this simple rule has asymptotic Pareto price of anarchy at most $1.283$ and, more strongly, supplies the individual guarantees behind that bound.

Third, the auction supplies an endogenous procedure for allocating opportunities above the guaranteed levels. In the general prior-free GUM construction, a universally feasible guarantee vector may leave profile-specific slack relative to attainable welfare, and the framework does not select a canonical rule for distributing that slack \citep{Csoka2026}. Here bids determine the allocation round by round.

Fourth, specialization permits sharper finite-horizon guarantees. The general prior-free NTU lifting theorem has a leading loss of order $\sqrt T\,\log(T/\alpha_i)$ under bounded effects \citep[Theorem~5.7]{Csoka2026}. In the bounded independent-stationary single-good model, the present results give an $O_D(\sqrt T)$ lower error for every bounded distribution and an $O_D(\log T)$ error for tame distributions, uniformly in the initial share.

Finally, the virtual-budget auction is a natural component for models with capped or lossy real transfers. Such hybrid mechanisms interpolate between pure NTU allocation and fully transferable utility.

\subsection{Relation to prior work}

The closest conceptual line is dynamic mechanism design through guaranteed utilities. \citet{AtheySegal2013} construct an efficient budget-balanced dynamic mechanism under Bayesian incentive constraints, and \citet{CsokaEtAl2024} construct a dynamic mechanism in which each player can secure a prescribed utility even against collusive behavior of the others. Guaranteed Utility Equilibrium formalizes the case in which the vector of maximin utilities is Pareto efficient; then all Nash equilibria have that same utility vector \citep[Proposition~2.1 and Lemma~3.9]{CsokaPongraczRodivilov2025}. Its approximate version gives corresponding bounds for approximate equilibria \citep[Lemma~3.7]{CsokaPongraczRodivilov2025}. The present paper uses only the elementary first half of this logic---an equilibrium utility cannot fall below a player's security level---and proves it directly where needed. The exact Pareto-efficiency condition of GUE is replaced here by the factor $1.283$.

The prior-free extension \citep{Csoka2026} lifts universally feasible targets through general TU-GUM and NTU-GUM constructions. The present paper implements the same fair floor through a fixed first-price credit rule without a prior-reporting stage.

A broader literature studies allocation when money is absent or constrained. Repetition can relax incentive constraints by linking independent decisions \citep{JacksonSonnenschein2007}. A direct precursor is \citet{GuoConitzerReeves2009}, who study repeated single-item allocation without monetary transfers through an artificial-payment system. Prior-free payment-free allocation was studied, among other settings, by \citet{GuoConitzer2010}, and \citet{BalseiroGurkanSun2019} developed multiagent mechanisms without money. Artificial-credit mechanisms for repeated allocation appear in \citet{GorokhBanerjeeIyer2016,GorokhBanerjeeIyer2021}, while fake-money competitive equilibria have been implemented at scale in course allocation \citep{BudishEtAl2017}.

The closest algorithmic work concerns robust repeated resource sharing. \citet{BanerjeeFikiorisTardos2023} study first-price pseudo-markets with artificial credits for reusable resources and guarantee a constant fraction of an ideal quota-constrained utility. \citet{FikiorisBanerjeeTardos2025} analyze a dynamic max-min rule, and \citet{LinEtAl2025} obtain a stronger constant guarantee for repeated first-price bidding with artificial currency through randomized bids. \citet{LinEtAlSODA2026} construct robust equilibria in the same broad i.i.d. repeated single-item setting, while \citet{LinEtAlITCS2026} use competitive subsidies to break the $0.6$ robustness barrier for repeated first-price auctions with simple strategies. The present paper uses a different benchmark and strategy class: it targets the full fair-floor value, lets the guaranteeing policy react to the current budget share and remaining horizon, and evaluates security against an arbitrarily coordinated coalition.

\section{Model, benchmark, and main theorem}\label{sec:model}

\subsection{The repeated virtual-budget mechanism}\label{subsec:mechanism}

There are $n\ge2$ players and $T$ independent rounds. In each round a single good is allocated by a first-price auction. Player $i$ has initial virtual budget $B_{i,0}>0$ and initial budget share
\[
\alpha_i=\frac{B_{i,0}}{\sum_j B_{j,0}},
\qquad
\sum_i\alpha_i=1.
\]
In round $t$, player $i$ privately observes a nonnegative value
\[
V_{i,t}\sim D_i,
\]
independently across players and rounds. Players submit bids constrained by their remaining virtual budgets. The highest bidder receives the good, obtains utility equal to her value, and pays her bid in virtual money. Ties among highest bidders are broken uniformly at random.\footnote{By adding an arbitrarily small continuous perturbation to her bids, each player can make the probability of a tie involving her zero. We therefore omit a separate analysis of ties.} The winning bid is burned. The remaining virtual money has no direct utility; it only constrains future bids. We assume that the winner, the paid bid, and the remaining virtual budgets are publicly observed after each round. For the equilibrium statements, we use the standard Bayesian interpretation in which the share vector, the distribution profile, the horizon, and the mechanism are common knowledge. The playerwise guarantee theorem itself requires player $i$ to know only her own distribution $D_i$.

The game is homogeneous in money. If after any history all remaining budgets are multiplied by the same positive constant, and all future bids are scaled by the same constant, then feasibility, winners, and utilities are unchanged. Thus after every round with positive total remaining money we may normalize the total remaining virtual money to be $1$. If total remaining money ever becomes $0$, no positive future bids are feasible; one may specify any fixed tie-breaking continuation. This degenerate case is never reached in the interior normalized recursion used below, and the endpoint cases are treated separately.

\subsection{The fair-floor function and its static origin}\label{subsec:static-origin}

\begin{definition}[Fair-floor target]\label{def:fair-floor}
For a value distribution $D$ and share $\alpha\in[0,1]$, let $F_D$ denote the CDF and let
\[
Q_D(u)=\inf\{x\ge0:F_D(x)\ge u\},\qquad u\in(0,1),
\]
be the generalized quantile; its endpoint values are immaterial. For $\alpha\in(0,1]$, using an independent $Z\sim U[0,1]$ and the randomized CDF-rank
\[
R_D(V,Z)=F_D(V^-)+Z\bigl(F_D(V)-F_D(V^-)\bigr),
\]
we write
\[
\phi(\alpha,D)=\E\left[V\,R_D(V,Z)^{1/\alpha-1}\right].
\]
Equivalently, if $Q_D$ is the generalized quantile function, then for $\alpha\in(0,1]$,
\begin{equation}\label{eq:phi-def}
\phi(\alpha,D)=\int_0^1 Q_D(u)u^{1/\alpha-1}\,du.
\end{equation}
At $\alpha=1$, the rank formula is interpreted as $\E[V]$ (the value of $R^0$ on the null event $R=0$ is irrelevant). We set $\phi(0,D)=0$. If $D$ is integrable, then $\phi(1,D)=\E[V]$.
\end{definition}

\begin{lemma}[Randomized-rank coupling]\label{lem:rank-coupling}
Let $V\sim D$, let $Z\sim U[0,1]$ be independent, and put $R=R_D(V,Z)$. Then
\[
R\sim U[0,1],
\qquad
Q_D(R)=V\quad\text{almost surely}.
\]
More generally, for every nonnegative or integrable Borel function $h$,
\begin{equation}\label{eq:rank-coupling}
\E[h(V,R)]=\int_0^1 h(Q_D(u),u)\,du.
\end{equation}
In particular, the two formulas in Definition \ref{def:fair-floor} agree, including when $D$ has atoms.
\end{lemma}

\begin{proof}
Write $F=F_D$. Fix $u\in(0,1)$, let $q=Q_D(u)$, and put $a=F(q^-)$ and $b=F(q)$. The generalized-inverse inequalities give $a\le u\le b$. On $\{V<q\}$ one has $R\le a$, while on $\{V>q\}$ one has $R\ge b$. Conditional on $V=q$, the variable $R$ is uniform on $[a,b]$. Consequently,
\[
\Pr(R\le u)
=a+(b-a)\frac{u-a}{b-a}=u,
\]
with the evident convention when $a=b$. The endpoint values follow trivially, so $R$ is uniform.

For every real $x$ and $u\in[0,1]$,
\[
\Pr(V\le x,\ R\le u)=\min\{F(x),u\}.
\]
Indeed, if $u\ge F(x)$ then $V\le x$ implies $R\le F(V)\le F(x)\le u$; if $u<F(x)$, then $V>x$ implies $R\ge F(V^-)\ge F(x)>u$, so the event $\{R\le u\}$ already forces $V\le x$. If $U\sim U[0,1]$, the generalized-inverse identity $Q_D(U)\le x$ if and only if $U\le F(x)$ shows that
\[
\Pr(Q_D(U)\le x,\ U\le u)=\min\{F(x),u\}.
\]
Hence $(V,R)$ and $(Q_D(U),U)$ have the same joint distribution. Formula \eqref{eq:rank-coupling} follows, and so do $Q_D(R)=V$ almost surely and the asserted equivalence of the two fair-floor formulas.
\end{proof}

\begin{proposition}[Static weighted-rank origin of the fair floor]\label{prop:static-origin}
Let $\alpha_1,\ldots,\alpha_n>0$ satisfy $\sum_i\alpha_i=1$. Draw values $V_i\sim D_i$ independently, construct independent randomized ranks
\[
R_i=R_{D_i}(V_i,Z_i)\sim U[0,1],
\]
and allocate the object to the player with the largest score $R_i^{1/\alpha_i}$. Then player $i$ receives the object with probability $\alpha_i$, and her expected utility is exactly
\[
\Pr(i\text{ receives the object})=\alpha_i,
\qquad
\E\left[V_i\one_{\{i\text{ receives the object}\}}\right]
=\phi(\alpha_i,D_i).
\]
Consequently, the vector $(\phi(\alpha_i,D_i))_{i=1}^n$ is jointly feasible ex ante in the one-shot allocation problem for every profile $(D_i)_i$ and every share vector $(\alpha_i)_i$.

If $\alpha_i=1/n$ and $D_i=D$ for every $i$, the weighted-rank rule is ex post efficient and
\begin{equation}\label{eq:symmetric-replica}
\phi(1/n,D)=\frac1n\E\left[\max_{1\le j\le n}V_j\right].
\end{equation}
\end{proposition}

\begin{proof}
Conditional on $R_i=r$, player $i$ wins precisely when
\[
R_j^{1/\alpha_j}<r^{1/\alpha_i}
\quad\text{for every }j\ne i.
\]
The score distributions are continuous, so ties have probability zero. Independence and uniformity of the randomized ranks give
\[
\Pr(i\text{ wins}\mid R_i=r)
=\prod_{j\ne i}r^{\alpha_j/\alpha_i}
=r^{(1-\alpha_i)/\alpha_i}
=r^{1/\alpha_i-1}.
\]
Using Lemma \ref{lem:rank-coupling},
\[
\E\left[V_i\one_{\{i\text{ wins}\}}\right]
=\int_0^1Q_{D_i}(r)r^{1/\alpha_i-1}\,dr
=\phi(\alpha_i,D_i).
\]
The same calculation without the value factor gives
\[
\Pr(i\text{ wins})=\int_0^1r^{1/\alpha_i-1}\,dr=\alpha_i.
\]
This proves joint feasibility and the expected-share identity. In the symmetric case all score transformations are the same, so the highest randomized rank wins. Since $Q_D$ is nondecreasing and $Q_D(R_j)=V_j$ almost surely, that player has a maximal value. The rule is therefore ex post efficient; symmetry gives each player a $1/n$ share of its expected welfare, proving \eqref{eq:symmetric-replica}.
\end{proof}

\subsection{Pareto price of anarchy}\label{subsec:pareto-poa}

\begin{definition}[Pareto price of anarchy]\label{def:pareto-poa}
Fix a one-round instance $(\alpha_N,D_N)$, and let $\mathcal U(\alpha_N,D_N)$ be the set of expected-utility vectors attainable by arbitrary measurable, possibly randomized, quota-free allocation rules. For $\lambda\ge1$, a vector $x\in\mathbb R_+^N$ has \emph{Pareto price of anarchy at most $\lambda$} if there is no $u\in\mathcal U(\alpha_N,D_N)$ such that
\[
u_i>\lambda x_i
\qquad\text{for every }i\in N.
\]
Equivalently, no feasible vector improves every coordinate of $x$ by a common factor strictly larger than $\lambda$. This comparison is invariant under separate positive rescalings of the players' utility units.

The repeated mechanism has \emph{asymptotic Pareto price of anarchy at most $\lambda$} if every accumulation point of normalized expected-utility vectors of Bayesian Nash equilibria has Pareto price of anarchy at most $\lambda$ in the corresponding one-round feasible region.
\end{definition}

\begin{proposition}[The $1.283$ fair-floor bound]\label{prop:1283-benchmark}
Assume that the distributions $D_i$ are nonnegative and integrable, and put
\[
g_i=\phi(\alpha_i,D_i).
\]
There is no $u\in\mathcal U(\alpha_N,D_N)$ such that
\[
u_i>1.283\,g_i
\qquad\text{for every }i\in N.
\]
Thus the fair-floor vector has Pareto price of anarchy at most $1.283$ relative to the quota-free feasible region.
\end{proposition}

\begin{proof}
This is \citet[Theorem~6.2]{Csoka2026}.
\end{proof}

\begin{remark}[Fixed-instance efficiency and universal optimality]\label{rem:maximality}
Proposition \ref{prop:1283-benchmark} compares the fair-floor vector with the quota-free feasible utility region of a fixed instance. In the corresponding repeated single-good environment, \citet[Proposition~4.8]{Csoka2026} proves the cross-instance statement that
\[
f^*(\alpha,D)=T\phi(\alpha,D)
\]
is Pareto-optimal among universally feasible target functions. The auction theorem below implements this target, while Proposition \ref{prop:1283-benchmark} converts the guarantee into the price-of-anarchy bound.
\end{remark}

\begin{remark}[Three elementary benchmarks]\label{rem:phi-examples}
The formula specializes as follows.
\begin{enumerate}[label=\textup{(\roman*)}, leftmargin=2.2em]
\item If $V=v$ almost surely, then $\phi(\alpha,D)=\alpha v$. With no variation in urgency, the benchmark reduces to the proportional share.
\item If $V\sim U[0,1]$, then $\phi(\alpha,D)=\alpha/(1+\alpha)$. In particular, $\phi(1/n,D)=1/(n+1)$, the $1/n$ share of the expected maximum $n/(n+1)$.
\item If $V=v$ with probability $p$ and $V=0$ otherwise, then
\[
\phi(\alpha,D)=\alpha v\bigl(1-(1-p)^{1/\alpha}\bigr).
\]
This illustrates the upper-rank emphasis: a small-share player concentrates her entitlement on the rare valuable rounds rather than receiving merely $\alpha\E[V]$.
\end{enumerate}
\end{remark}

\subsection{Main theorem}\label{subsec:target}

\begin{theorem}[Main fair-floor guarantee]\label{thm:main-guarantee}
Fix a player $i$ in the repeated first-price virtual-budget mechanism and assume that $D_i$ is nonnegative and integrable. There is a nonnegative function
\[
\rho_{D_i}(T)=o_{D_i}(T),
\]
independent of the initial share $\alpha_i$, such that for every horizon $T$ player $i$ has a behavioral strategy satisfying
\[
\E[U_i^T]\ge T\phi(\alpha_i,D_i)-\rho_{D_i}(T)
\]
against every joint, possibly collusive, strategy of the other players. The strategy depends on $D_i$, the remaining horizon, the current budget share, the current value, and private randomization, but not on the other players' distributions.

Write
\[
\phi_i(\alpha)=\phi(\alpha,D_i),
\qquad
H_T=\sum_{m=1}^T\frac1m.
\]
The following explicit choices are available:
\begin{enumerate}[label=\textup{(\roman*)}, leftmargin=2.4em]
\item if $D_i$ satisfies Assumption~\ref{ass:tame}, then with $M_i=\|Q_{D_i}\|_\infty$, $\gamma_i=\phi_i'(1)>0$, $K_i=M_i/\gamma_i$, and $L_i=\sup_{\alpha\in[0,1]}|\phi_i''(\alpha)|$,
\[
\rho_{D_i}(T)=\bigl(2M_iK_i+2L_iK_i^2\bigr)H_T;
\]
\item if $D_i$ is bounded and nonzero, then with $M_i=\|Q_{D_i}\|_\infty$ and $\gamma_i=\phi_i'(1)>0$,
\[
\rho_{D_i}(T)=\frac{40M_i^2}{\gamma_i}\sqrt T;
\]
\item if $\E[V_i^{p_i}]<\infty$ for some $p_i>1$, then one may take
\[
\rho_{D_i}(T)=O_{D_i}\!\left(T^{(p_i+3)/(2p_i+2)}\right).
\]
For $D_i\equiv0$, the zero strategy gives $\rho_{D_i}=0$.
\end{enumerate}

In the two-player Bidder--Adversary game, the corresponding upper value satisfies
\[
\upC_T(\alpha,D)\le T\phi(\alpha,D)+O_D(\log T)
\]
for every bounded nonnegative $D$, and $\upC_T(\alpha,D)\le T\phi(\alpha,D)+o_D(T)$ for every integrable nonnegative $D$. Thus $T\phi$ is the leading term of both auxiliary values, with the stated lower-side smoothness distinction.
\end{theorem}

\begin{proof}
The tame and bounded original-game guarantees are Corollaries~\ref{cor:regular-guarantee} and \ref{cor:bounded-guarantee}. The integrable and finite-moment statements follow from Corollaries~\ref{cor:integrable-guarantee} and \ref{cor:finite-moment-rate}. The auxiliary upper estimates are Theorem~\ref{thm:upper-log} and Corollary~\ref{cor:integrable-guarantee}. All of these results are proved below from the explicit strategy-transfer Lemma~\ref{lem:aggregation}.
\end{proof}

\begin{corollary}[Price of anarchy]\label{cor:1283-poa}
Fix a finite population, a positive initial share vector, and nonnegative integrable stationary value distributions. Let $T_k\to\infty$, and for each $k$ let $s^{T_k}$ be a $\delta_{T_k}$-Bayesian Nash equilibrium of the $T_k$-round virtual first-price auction, where $\delta_{T_k}=o(T_k)$. Write
\[
x_i^{T_k}=\frac{1}{T_k}\E\bigl[U_i^{T_k}(s^{T_k})\bigr].
\]
Every accumulation point $x$ of $(x^{T_k})_k$ has Pareto price of anarchy at most $1.283$. In particular, the same conclusion holds for every sequence of exact Bayesian Nash equilibria.
\end{corollary}

\begin{proof}
Put $g_i=\phi(\alpha_i,D_i)$. By Theorem~\ref{thm:main-guarantee}, player $i$ has a strategy guaranteeing
\[
T_k g_i-\rho_{D_i}(T_k),
\qquad
\rho_{D_i}(T_k)=o(T_k),
\]
against every strategy of the others. The $\delta_{T_k}$-equilibrium condition therefore gives
\[
x_i^{T_k}
\ge
 g_i-\frac{\rho_{D_i}(T_k)+\delta_{T_k}}{T_k}.
\]
Hence every accumulation point satisfies $x_i\ge g_i$ for all $i$.

For completeness, each normalized equilibrium vector $x^{T_k}$ is itself attainable in the one-round feasible region. Choose a round uniformly from $\{1,\ldots,T_k\}$, sample the equilibrium history before that round, and apply the equilibrium allocation rule to a fresh current value profile. Current values are independent of the sampled history and have distribution $\prod_iD_i$, so averaging over the sampled history and round defines a measurable randomized one-round allocation rule with expected-utility vector $x^{T_k}$.

If some feasible vector $u$ satisfied $u_i>1.283\,x_i$ for every $i$, then $u_i>1.283\,g_i$ for every $i$, contradicting Proposition~\ref{prop:1283-benchmark}. This proves the claim.
\end{proof}

\section{Reduction to the Bidder--Adversary game}\label{sec:bidder-adversary}

Fix one player $i$ whose guarantee we study and call her \Bidder{}. The auxiliary game replaces all remaining players by a single \Adversary{} that controls their combined virtual budget and minimizes the \Bidder{}'s total utility. Write
\[
\alpha=\alpha_i,
\qquad
D=D_i.
\]

\begin{definition}[The Bidder--Adversary game $Z_T(\alpha,D)$]
The game $Z_T(\alpha,D)$ is a $T$-round zero-sum game between \Bidder{} and \Adversary{}. Initially,
\[
B_0^{\BidderTag}=\alpha,
\qquad
B_0^{\AdversaryTag}=1-\alpha.
\]
In each round $t$, \Bidder{} observes a private value $V_t\sim D$, independent of the past. \Adversary{} does not observe $V_t$. They simultaneously submit feasible bids
\[
b_t\in[0,B_{t-1}^{\BidderTag}],
\qquad
 e_t\in[0,B_{t-1}^{\AdversaryTag}].
\]
If $b_t>e_t$, then \Bidder{} wins, gains utility $V_t$, and pays $b_t$. If $e_t\ge b_t$, then \Adversary{} wins and pays $e_t$. Ties are broken against \Bidder{}. After the round, the paid bid and the resulting normalized shares are public. \Bidder{}'s payoff is
\[
U^{\BidderTag}=\sum_{t=1}^T V_t\one_{\{b_t>e_t\}},
\]
and \Adversary{}'s payoff is $-U^{\BidderTag}$.

Strategies are behavioral. At every history, \Bidder{} may randomize after observing $V_t$, while \Adversary{} may randomize independently of $V_t$. For distributions with atoms, \Bidder{} draws an independent $Z_t\sim U[0,1]$ and implements a rank-based bid $b(R_D(V_t,Z_t))$; for atomless distributions the auxiliary randomization can be omitted.
\end{definition}

We use the lower and upper values
\begin{align*}
\lowC_T(\alpha,D)
&=
\sup_{\sigma^{\BidderTag}}\inf_{\sigma^{\AdversaryTag}}
\E_{\sigma^{\BidderTag},\sigma^{\AdversaryTag}}\left[\sum_{t=1}^T V_t\one_{\{b_t>e_t\}}\right],\\
\upC_T(\alpha,D)
&=
\inf_{\sigma^{\AdversaryTag}}\sup_{\sigma^{\BidderTag}}
\E_{\sigma^{\BidderTag},\sigma^{\AdversaryTag}}\left[\sum_{t=1}^T V_t\one_{\{b_t>e_t\}}\right].
\end{align*}
The lower value governs the guarantee, while the upper value identifies the asymptotic leading term of the auxiliary game.

\begin{lemma}[Aggregation reduction]\label{lem:aggregation}
Every behavioral strategy of \Bidder{} in $Z_T(\alpha_i,D_i)$ has an implementation by player $i$ in the original $n$-player game with the following property: against every joint, possibly collusive, strategy of the other players, the focal player's payoff has the same law as the auxiliary payoff against some feasible \Adversary{} strategy. Consequently:
\begin{enumerate}[label=\textup{(\roman*)}, leftmargin=2.2em]
\item any explicit auxiliary strategy that guarantees $G$ transfers to an original-game strategy that guarantees $G$ exactly;
\item for every $\varepsilon>0$, player $i$ has a strategy guaranteeing
\[
\E[U_i^T]\ge \lowC_T(\alpha_i,D_i)-\varepsilon.
\]
\end{enumerate}
\end{lemma}

\begin{proof}
Fix a joint strategy of the players other than $i$, and work in normalized units. Let $S_t>0$ be the total remaining virtual budget immediately before round $t$ in the original game. An auxiliary bid $x$ is implemented by the raw bid $S_t x$, while a raw opposing bid is represented by its ratio to $S_t$. The auxiliary \Adversary{} privately simulates the opponents' individual value draws, private randomizations, private histories, raw remaining budgets, and the scale $S_t$. In round $t$ the simulation produces their simultaneous raw bids, and \Adversary{} submits
\[
e_t=\frac{1}{S_t}\max_{j\ne i}b_{j,t}.
\]
This normalized bid is feasible because the maximum raw opposing bid is at most the opponents' aggregate remaining budget. If player $i$ wins, the opposing aggregate budget is unchanged. If player $i$ loses, an opponent with the maximum bid wins and exactly $S_t e_t$ is burned from the opposing aggregate budget. Consequently the normalized aggregate shares evolve exactly as the two budgets in $Z_T(\alpha_i,D_i)$. The scalar $S_t$ is recursively determined by the public sequence of normalized winning bids and may also be retained in \Adversary{}'s private simulation, so the construction includes opponent strategies that depend on absolute, rather than only normalized, budget levels. If a zero-total history is reached, both games use the fixed degenerate continuation specified in the model.

The simulated opponents' current bids depend only on their information before the simultaneous move. By the independence assumption, this information is independent of player $i$'s current value, so the resulting $e_t$ is a legal \Adversary{} action. The auxiliary private state may retain the simulated individual budgets and histories; it need not be publicly revealed. Conversely, player $i$ can implement any strategy of \Bidder{} in the original game while ignoring the additional public information about the individual opponents.

Therefore, for every strategy of \Bidder{} and every fixed joint opponent strategy, the focal player's bids, wins, and utilities have the same law as against the corresponding auxiliary \Adversary{} strategy. This proves the strategy-transfer statement and part \textup{(i)}. Part \textup{(ii)} follows by choosing an auxiliary strategy whose worst-case payoff is within $\varepsilon$ of the supremum defining $\lowC_T(\alpha_i,D_i)$.
\end{proof}

\begin{remark}[Role of independence]
The reduction uses that player $i$'s current and future values are independent of the other side's information and of past values. With intertemporal correlation, the state would have to include posterior or Markov information; the present one-dimensional recursion would no longer be exact.
\end{remark}

\section{Proof ingredients}\label{sec:ingredients}

\subsection{One-step dynamics and barrier propagation}\label{sec:barrier-propagation}

Normalize total remaining money to $1$ after each round. If \Bidder{}'s current share is $\alpha\in(0,1)$ and $m$ rounds remain, then after \Bidder{} wins with bid $b$ the new share is
\begin{equation}\label{eq:alpha-plus}
\alpha^+(b)=\frac{\alpha-b}{1-b},
\end{equation}
and after \Adversary{} wins with bid $e$ the new share is
\begin{equation}\label{eq:alpha-minus}
\alpha^-(e)=\frac{\alpha}{1-e}.
\end{equation}

The randomized-rank coupling allows every current value to be represented as $Q_D(r)$ with $r\sim U[0,1]$. For a bounded Borel continuation function $g:[0,1]\to\mathbb R$, a Borel rank-bid rule $\beta:[0,1]\to[0,\alpha]$, and a pure feasible \Adversary{} bid $e\in[0,1-\alpha]$, define the one-step operator
\begin{align}\label{eq:one-step-operator}
\mathcal J_g(\alpha;\beta,e)
:=\int_0^1\bigg[&
\one_{\{\beta(r)>e\}}
\Bigl(Q_D(r)+g\bigl(\alpha^+(\beta(r))\bigr)\Bigr)
\nonumber\\
&+\one_{\{\beta(r)\le e\}}
 g\bigl(\alpha^-(e)\bigr)
\bigg]dr.
\end{align}
For a \Adversary{} bid distribution $\lambda$ on $[0,1-\alpha]$, write
\[
\mathcal J_g(\alpha;\beta,\lambda)
=\int \mathcal J_g(\alpha;\beta,e)\,\lambda(de).
\]
The endpoint states are handled by the explicit actions stated below; no normalized transition formula is needed when one side has zero budget.

\begin{lemma}[Barrier propagation]\label{lem:barrier-propagation}
Let $L_m,U_m:[0,1]\to\mathbb R$ be bounded Borel functions with $L_0\le0\le U_0$.

\begin{enumerate}[label=\textup{(\alph*)}, leftmargin=2.2em]
\item Suppose that for every $m\ge1$ there is a jointly Borel map $(\alpha,r)\mapsto\beta_m(\alpha,r)$ on $(0,1)\times[0,1]$, with $\beta_m(\alpha,r)\in[0,\alpha]$, such that for every interior state $\alpha$
\begin{equation}\label{eq:lower-barrier-condition}
\inf_{e\in[0,1-\alpha]}
\mathcal J_{L_{m-1}}(\alpha;\beta_m(\alpha,\cdot),e)
\ge L_m(\alpha),
\end{equation}
and suppose the corresponding one-step inequalities hold for specified endpoint actions. Then the recursively concatenated Markov behavioral policy guarantees $L_m(\alpha)$ from every state. In particular,
\[
\lowC_m(\alpha,D)\ge L_m(\alpha).
\]

\item Suppose that for every $m\ge1$ and interior state $\alpha$ there is a Borel probability kernel $\lambda_m(\alpha,de)$ on feasible \Adversary{} bids such that
\begin{equation}\label{eq:upper-barrier-condition}
\sup_{\beta}
\mathcal J_{U_{m-1}}(\alpha;\beta,\lambda_m(\alpha,\cdot))
\le U_m(\alpha),
\end{equation}
where the supremum is over Borel rank-bid rules $\beta:[0,1]\to[0,\alpha]$, and suppose the analogous endpoint inequalities hold. Then the recursively concatenated \Adversary{} policy holds every \Bidder{} strategy to at most $U_m(\alpha)$. In particular,
\[
\upC_m(\alpha,D)\le U_m(\alpha).
\]
\end{enumerate}
\end{lemma}

\begin{proof}
Both statements follow by induction on the remaining horizon, conditional on the current public history.

For the lower statement, condition also on \Adversary{}'s private information immediately before the simultaneous move. The current rank is independent of that information. Conditional on any realized pure bid $e$, the induction hypothesis bounds the continuation payoff from below by $L_{m-1}$ at the public next state, so \eqref{eq:lower-barrier-condition} gives the desired bound. Averaging over any private randomization of \Adversary{} preserves it.

For the upper statement, represent the current behavioral bid kernel by a Borel function of the current rank and an independent uniform seed, and condition on \Bidder{}'s private information before the current value is drawn and on that seed. This leaves a Borel rank-bid rule $\beta$. The induction hypothesis bounds the continuation payoff from above by $U_{m-1}$ even though \Bidder{} retains the enlarged private history; future values remain independent. Inequality \eqref{eq:upper-barrier-condition} therefore applies. Averaging over \Bidder{}'s private randomization completes the induction. The endpoint arguments are identical once their actions are specified.
\end{proof}

\begin{remark}[Constructive dynamic induction]\label{rem:no-selection}
All later policies are explicit state-dependent rules. Lemma \ref{lem:barrier-propagation} propagates their one-step bounds directly, without a Bellman equality, measurable selectors, or finite-horizon minimax.
\end{remark}

\subsection{Analytic properties of the fair-floor curve}\label{subsec:analytic}

The proofs use two levels of regularity. Concavity and a quantitative interior curvature bound hold for every bounded distribution; a $C^2$ extension to $\alpha=0$ gives the sharper logarithmic lower barrier.

\begin{lemma}[Concavity of the fair-floor curve]\label{lem:concavity}
If $Q$ is bounded, nonnegative, and nondecreasing, then $\alpha\mapsto\phi(\alpha,D)$ is increasing and concave on $(0,1]$.
\end{lemma}

\begin{proof}
Approximate $Q$ from below by nonnegative nondecreasing step functions. For a step function, a constant component contributes a linear term in $\alpha$, while a step increment beginning at $s\in(0,1)$ contributes
\[
g_s(\alpha)=\int_s^1 u^{1/\alpha-1}\,du=\alpha(1-s^{1/\alpha}).
\]
A direct calculation gives
\[
g_s''(\alpha)=-\frac{s^{1/\alpha}(\log s)^2}{\alpha^3}\le0.
\]
Thus every step approximation gives an increasing concave fair-floor curve. Monotone convergence then passes both monotonicity and concavity to $\phi$.
\end{proof}

\begin{lemma}[Basic analytic bounds for bounded distributions]\label{lem:bounded-analytic}
Let $D$ be bounded and nonnegative, let $Q=Q_D$, and put $M=\|Q\|_\infty$. For $\alpha\in(0,1]$, the function
\[
\phi(\alpha)=\int_0^1 Q(u)u^{1/\alpha-1}\,du
\]
is $C^2$ on $(0,1]$ and satisfies
\begin{align}
0\le \phi(\alpha)&\le M\alpha,\label{eq:bounded-phi-linear}\\
\phi'(\alpha)&=\frac1{\alpha^2}\int_0^1 Q(u)(-\log u)u^{1/\alpha-1}\,du,\label{eq:bounded-phi-prime}\\
|\phi''(\alpha)|&\le \frac{4M}{\alpha}.\label{eq:bounded-phi-second}
\end{align}
Moreover, if $D$ is not identically zero, then
\begin{equation}\label{eq:gamma-bounded}
\gamma_D:=\inf_{\alpha\in(0,1]}\phi'(\alpha)
=\phi'(1)
=\int_0^1 Q(u)(-\log u)\,du
>0.
\end{equation}
\end{lemma}

\begin{proof}
The bound \eqref{eq:bounded-phi-linear} follows from
\[
\int_0^1 u^{1/\alpha-1}\,du=\alpha.
\]
Dominated differentiation on every compact subinterval of $(0,\infty)$ gives the derivative formula and continuity of the resulting derivatives. For the second derivative, write
\[
A_k(\alpha)=\int_0^1 Q(u)(-\log u)^k u^{1/\alpha-1}\,du.
\]
Then
\[
\phi''(\alpha)=-\frac{2}{\alpha^3}A_1(\alpha)+\frac1{\alpha^4}A_2(\alpha).
\]
Since $A_1(\alpha)\le M\alpha^2$ and $A_2(\alpha)\le 2M\alpha^3$, we get \eqref{eq:bounded-phi-second}.

Finally, Lemma \ref{lem:concavity} and differentiability imply that $\phi'$ is nonincreasing on $(0,1]$. Hence its infimum is $\phi'(1)$. Formula \eqref{eq:bounded-phi-prime} at $\alpha=1$ gives the displayed integral, which is strictly positive whenever $D$ is not identically zero.
\end{proof}

For the logarithmic lower bound we use the following regularity condition.

\begin{assumption}[Tame bounded distribution]\label{ass:tame}
The distribution $D$ is nonnegative, bounded, and not identically zero, and the fair-floor function
\[
\phi(\alpha)=\int_0^1 Q_D(u)u^{1/\alpha-1}\,du,
\qquad
\phi(0)=0,
\]
extends to a $C^2$ function on $[0,1]$. We use the associated constants
\begin{equation}\label{eq:tame-bounds}
M=\|Q_D\|_\infty,
\qquad
\gamma=\inf_{\alpha\in(0,1]}\phi'(\alpha),
\qquad
L=\sup_{\alpha\in[0,1]}|\phi''(\alpha)|.
\end{equation}
Lemma \ref{lem:bounded-analytic} shows that $\gamma>0$; compactness gives $L<\infty$. Thus the derivative bounds follow from the $C^2$ extension. Flat quantile intervals are allowed and are handled by the push-forward \Adversary{} distribution in Lemma \ref{lem:static-saddle}.
\end{assumption}

\subsection{The static tangent game}\label{sec:static-tangent}

The tangent game is the first-order local problem associated with the candidate continuation value $m\phi(\alpha)$. Set
\[
\mu_\alpha=(1-\alpha)\phi'(\alpha),
\qquad
b=\frac{B}{m\mu_\alpha},
\qquad
e=\frac{Y}{m\mu_\alpha}.
\]
For bids of order $1/m$, the share transitions satisfy
\[
(m-1)\phi\bigl(\alpha^+(b)\bigr)
=(m-1)\phi(\alpha)-B+O(1/m),
\]
and
\[
(m-1)\phi\bigl(\alpha^-(e)\bigr)
=(m-1)\phi(\alpha)+\frac{\alpha}{1-\alpha}Y+O(1/m).
\]
After subtracting the common continuation term, the one-step payoff is therefore approximately $Q(r)-B$ when \Bidder{} wins and $\alpha Y/(1-\alpha)$ when \Adversary{} wins. The following one-shot game solves this local problem exactly; Lemma~\ref{lem:dynamic-tangent} later makes the approximation error precise.

Let $Q=Q_D$ be bounded, nonnegative, and nondecreasing, put $M=\|Q\|_\infty$, and fix $\alpha\in(0,1)$. Write
\[
a=\frac{1-\alpha}{\alpha},
\qquad
c=\frac{\alpha}{1-\alpha}.
\]
The tangent game is a one-shot zero-sum game. \Bidder{} observes $r\sim U[0,1]$ and value $Q(r)$. \Bidder{} chooses a tangent bid $B\ge0$, and \Adversary{} chooses a tangent bid $Y\ge0$. If $B>Y$, \Bidder{}'s payoff is $Q(r)-B$. If $Y\ge B$, \Bidder{}'s payoff is $cY$.

Define, for $r\in(0,1]$,
\begin{equation}\label{eq:B-alpha}
B_\alpha(r)=
\frac{1-\alpha}{\alpha}
 r^{-1/\alpha}\int_0^r Q(s)s^{1/\alpha-1}\,ds,
\end{equation}
and set $B_\alpha(0)$ equal to the right limit. Then
\begin{equation}\label{eq:B-ode}
rB_\alpha'(r)=
\frac{1-\alpha}{\alpha}Q(r)-\frac1\alpha B_\alpha(r)
\end{equation}
where the derivative exists. Moreover,
\begin{equation}\label{eq:B-bound}
B_\alpha(r)
\le(1-\alpha)Q(r)
\le(1-\alpha)M.
\end{equation}

\begin{lemma}[Static tangent saddle]\label{lem:static-saddle}
Let $Q$ be bounded, nonnegative, and nondecreasing. In the static tangent game, the strategy $r\mapsto B_\alpha(r)$ guarantees \Bidder{} at least $\phi(\alpha,D)$.

Conversely, let $S_\alpha$ have distribution $\Pr(S_\alpha\le s)=s^a$ on $[0,1]$, and let \Adversary{} bid
\begin{equation}\label{eq:opponent-H}
Y=B_\alpha(S_\alpha).
\end{equation}
Then this \Adversary{} distribution holds \Bidder{}'s expected payoff to at most $\phi(\alpha,D)$. Hence the static tangent game has value $\phi(\alpha,D)$. If $B_\alpha$ is strictly increasing, this is equivalently described by $H_\alpha(B_\alpha(r))=r^a$; if $B_\alpha$ has flat intervals, the push-forward distribution has the corresponding atoms.
\end{lemma}

\begin{proof}
First record the elementary regularity of $B_\alpha$. On every interval $[\varepsilon,1]$, the primitive
\[
F(r)=\int_0^r Q(s)s^{1/\alpha-1}\,ds
\]
is absolutely continuous, and therefore $B_\alpha(r)=((1-\alpha)/\alpha)r^{-1/\alpha}F(r)$ is absolutely continuous there. The alternative representation
\[
B_\alpha(r)=(1-\alpha)\int_0^1 Q(rt)\frac1\alpha t^{1/\alpha-1}\,dt
\]
shows, by dominated convergence and the existence of the right limit of the monotone function $Q$ at $0$, that $B_\alpha$ is continuous at $0$. On $[\varepsilon,1]$, \eqref{eq:B-ode} and \eqref{eq:B-bound} imply $B_\alpha'\ge0$ almost everywhere, so $B_\alpha$ is nondecreasing on $[\varepsilon,1]$ and hence on $[0,1]$ by letting $\varepsilon\downarrow0$. The functions differentiated below are therefore absolutely continuous on compact subintervals of $(0,1)$, and the endpoint cases follow by continuity.

First consider \Bidder{}'s guarantee. For a \Adversary{} tangent bid $Y\le B_\alpha(1)$ define the upper cutoff
\[
t_+(Y)=\sup\{s\in[0,1]:B_\alpha(s)\le Y\},
\]
with $t_+(Y)=0$ if the set is empty. With ties broken against \Bidder{}, she wins exactly on the ranks $r>t_+(Y)$. If $Y$ belongs to the support interval of $B_\alpha$, continuity gives $B_\alpha(t_+(Y))=Y$; if $Y$ is below the support, then $t_+(Y)=0$ and the product $t_+(Y)Y$ is zero. Thus, in all cases $Y\le B_\alpha(1)$, \Bidder{}'s expected payoff from $B_\alpha(r)$ is at least
\[
\Psi(t)=
\int_t^1(Q(r)-B_\alpha(r))\,dr
+
\frac{\alpha}{1-\alpha}tB_\alpha(t),
\qquad t=t_+(Y).
\]
Using \eqref{eq:B-ode}, $\Psi'(t)=0$ at every differentiability point of $B_\alpha$. Hence $\Psi$ is constant. At $t=1$,
\[
\Psi(1)=\frac{\alpha}{1-\alpha}B_\alpha(1)
=\int_0^1 Q(s)s^{1/\alpha-1}\,ds
=\phi(\alpha).
\]
For $Y>B_\alpha(1)$, \Bidder{} loses for sure and receives the tangent continuation payoff $\frac{\alpha}{1-\alpha}Y\ge\frac{\alpha}{1-\alpha}B_\alpha(1)=\phi(\alpha)$. Hence \Bidder{} guarantees $\phi(\alpha)$.

Now let \Adversary{} use the push-forward bid $Y=B_\alpha(S_\alpha)$, where $S_\alpha$ has density $a s^{a-1}$. Fix a type $r$ of \Bidder{} and an arbitrary tangent bid $x\ge0$. Put
\[
t_-(x)=\sup\{s\in[0,1]: B_\alpha(s)<x\},
\]
with the endpoint conventions $t_-(x)=0$ below the support and $t_-(x)=1$ above the support. Since ties are lost by \Bidder{}, she wins only when $S_\alpha<t_-(x)$. If this winning event has positive probability, continuity gives $x\ge B_\alpha(t_-(x))$; below the support the winning probability is zero. Therefore her payoff from bid $x$ is at most
\[
P(r,t)=
t^a(Q(r)-B_\alpha(t))
+
\int_t^1 \frac{\alpha}{1-\alpha}B_\alpha(s)a s^{a-1}\,ds
\]
with $t=t_-(x)$. Thus it suffices to maximize $P(r,t)$ over $t\in[0,1]$.

At differentiability points of $B_\alpha$, using \eqref{eq:B-ode} gives
\[
\frac{\partial}{\partial t}P(r,t)=a t^{a-1}(Q(r)-Q(t)).
\]
Equivalently, by absolute continuity, for $t<r$,
\[
P(r,r)-P(r,t)=\int_t^r a s^{a-1}(Q(r)-Q(s))\,ds\ge0,
\]
and for $t>r$,
\[
P(r,t)-P(r,r)=\int_r^t a s^{a-1}(Q(r)-Q(s))\,ds\le0.
\]
Hence $P(r,t)$ is maximized at $t=r$, with possible flat intervals of maximizers when $Q$ is flat. Even if a flat segment of $B_\alpha$ prevents a single bid from inducing the cutoff $t=r$, $P(r,r)$ remains an upper bound on every feasible response. Thus the ex ante payoff is at most
\begin{align*}
\int_0^1 P(r,r)\,dr
&=\int_0^1 r^a Q(r)\,dr
 -\int_0^1 r^a B_\alpha(r)\,dr
 +\int_0^1\int_r^1 \frac{\alpha}{1-\alpha}B_\alpha(s)a s^{a-1}\,ds\,dr\\
&=\int_0^1 r^a Q(r)\,dr
 -\int_0^1 r^a B_\alpha(r)\,dr
 +\frac{\alpha}{1-\alpha}a\int_0^1 B_\alpha(s)s^a\,ds\\
&=\int_0^1 Q(r)r^a\,dr
=\phi(\alpha),
\end{align*}
because $a=1/\alpha-1$ and $(\alpha/(1-\alpha))a=1$. Thus \Adversary{} holds \Bidder{} down to $\phi(\alpha)$.
\end{proof}

\begin{lemma}[Normalization identity]\label{lem:mu}
Let
\[
\mu_\alpha=
\int_0^1 B_\alpha(r)\frac1\alpha r^{1/\alpha-1}\,dr.
\]
Then
\begin{equation}\label{eq:mu}
\mu_\alpha=(1-\alpha)\phi'(\alpha).
\end{equation}
Consequently, whenever a constant $\gamma>0$ satisfies $\phi'(\alpha)\ge\gamma$ on the relevant interval,
\begin{equation}\label{eq:scaled-bound}
\frac{B_\alpha(r)}{\mu_\alpha}
\le\frac{M}{\gamma}=:K.
\end{equation}
\end{lemma}

\begin{proof}
Using \eqref{eq:B-alpha} and Fubini,
\begin{align*}
\mu_\alpha
&=
\int_0^1
\frac{1-\alpha}{\alpha}
 r^{-1/\alpha}\left(\int_0^r Q(s)s^{1/\alpha-1}\,ds\right)
\frac1\alpha r^{1/\alpha-1}\,dr\\
&=\frac{1-\alpha}{\alpha^2}
\int_0^1 Q(s)s^{1/\alpha-1}\left(\int_s^1\frac{dr}{r}\right)ds\\
&=\frac{1-\alpha}{\alpha^2}
\int_0^1 Q(s)(-\log s)s^{1/\alpha-1}\,ds
=(1-\alpha)\phi'(\alpha).
\end{align*}
The bound \eqref{eq:scaled-bound} follows from \eqref{eq:B-bound}, \eqref{eq:mu}, and the chosen lower bound $\phi'\ge\gamma$ for $\alpha\in(0,1)$.
\end{proof}

\subsection{Dynamic--tangent comparison and policy measurability}\label{subsec:dynamic-tangent}

The next lemma records all one-step errors needed later. Its upper inequalities use only concavity. The logarithmic lower inequalities use the tame $C^2$ bound, while the robust lower inequalities use $|\phi''(\alpha)|\le4M/\alpha$ away from a small-share region.

\begin{lemma}[Dynamic--tangent comparison]\label{lem:dynamic-tangent}
Let $D$ be bounded, nonnegative, and not identically zero. Write $M=\|Q_D\|_\infty$, choose $\gamma>0$ with $\phi'(\alpha)\ge\gamma$ on $(0,1]$, and put $K=M/\gamma$. Fix $\alpha\in(0,1)$ and an integer $m$ such that $K/m\le1/2$. Let
\[
c=\frac{\alpha}{1-\alpha},
\qquad
\mu_\alpha=(1-\alpha)\phi'(\alpha),
\qquad
 e_{\max}=\frac{B_\alpha(1)}{m\mu_\alpha}.
\]
For a bid $b$ by \Bidder{} put $B=m\mu_\alpha b$, and for a bid $e$ by \Adversary{} put $Y=m\mu_\alpha e$.

\smallskip
\noindent
\textup{(i) Upper comparisons.} For every feasible $b$,
\begin{equation}\label{eq:dt-upper-win}
(m-1)\phi(\alpha^+(b))
\le (m-1)\phi(\alpha)-B+\frac{C_{\rm up}}m,
\end{equation}
where
\[
C_{\rm up}=2MK.
\]
For every feasible $e\le e_{\max}$,
\begin{equation}\label{eq:dt-upper-lose}
(m-1)\phi(\alpha^-(e))
\le (m-1)\phi(\alpha)+cY+\frac{C_{\rm up}}m.
\end{equation}

\smallskip
\noindent
\textup{(ii) Tame lower comparisons.} Suppose $|\phi''|\le L$ on $[0,1]$ and $\alpha\ge K/m$, and let $b=B_\alpha(r)/(m\mu_\alpha)$. Then the canonical bid is feasible, and for every feasible $e\le e_{\max}$,
\begin{align}
(m-1)\phi(\alpha^+(b))
&\ge (m-1)\phi(\alpha)-B_\alpha(r)-\frac{C_{\rm reg}}m,
\label{eq:dt-regular-win}\\
(m-1)\phi(\alpha^-(e))
&\ge (m-1)\phi(\alpha)+cY-\frac{C_{\rm reg}}m,
\label{eq:dt-regular-lose}
\end{align}
where
\[
C_{\rm reg}=2MK+2LK^2.
\]
If $e_{\max}\le1-\alpha$, then also
\begin{equation}\label{eq:dt-regular-above}
(m-1)\phi(\alpha^-(e_{\max}))
\ge m\phi(\alpha)-\frac{C_{\rm reg}}m.
\end{equation}

\smallskip
\noindent
\textup{(iii) Bounded lower comparisons.} Without a uniform $C^2$ assumption, suppose in addition that $m\ge16$ and $\alpha\ge K/\sqrt m$, and again let $b=B_\alpha(r)/(m\mu_\alpha)$. Then, for every feasible $e\le e_{\max}$,
\begin{align}
(m-1)\phi(\alpha^+(b))
&\ge (m-1)\phi(\alpha)-B_\alpha(r)-\frac{20MK}{\sqrt m},
\label{eq:dt-bounded-win}\\
(m-1)\phi(\alpha^-(e))
&\ge (m-1)\phi(\alpha)+cY-\frac{20MK}{\sqrt m}.
\label{eq:dt-bounded-lose}
\end{align}
If $e_{\max}\le1-\alpha$, then also
\begin{equation}\label{eq:dt-bounded-above}
(m-1)\phi(\alpha^-(e_{\max}))
\ge m\phi(\alpha)-\frac{20MK}{\sqrt m}.
\end{equation}
\end{lemma}

\begin{proof}
Since $\phi'(\alpha)\le M$ for bounded $D$, every admissible lower bound $\gamma$ satisfies $\gamma\le M$, and hence $K\ge1$. Set
\[
d_b=\frac{(1-\alpha)b}{1-b},
\qquad
d_e=\frac{\alpha e}{1-e},
\qquad
f_b=\frac{m-1}{m(1-b)},
\qquad
f_e=\frac{m-1}{m(1-e)}.
\]
Then $\alpha^+(b)=\alpha-d_b$, $\alpha^-(e)=\alpha+d_e$, and the normalization identity gives the exact relations
\begin{equation}\label{eq:dt-linear-identities}
(m-1)\phi'(\alpha)d_b=Bf_b,
\qquad
(m-1)\phi'(\alpha)d_e=cYf_e.
\end{equation}
Whenever $e\le e_{\max}$, one has $Y\le B_\alpha(1)$ and hence
\begin{equation}\label{eq:dt-static-bounds}
0\le cY\le cB_\alpha(1)=\phi(\alpha)\le M.
\end{equation}
For the canonical \Bidder{} bid, $B=B_\alpha(r)\le M$ and $b\le K/m$.

The elementary coefficient errors are as follows. For an arbitrary feasible $b$, the positive part of $B(1-f_b)$ can occur only when $b<1/m$; using $\mu_\alpha\le M$ gives
\[
[B(1-f_b)]_+
\le \frac{B}{m(1-b)}
\le \frac{Mb}{1-b}
\le \frac{2M}{m}.
\]
For the canonical bid, $b\le K/m\le1/2$, so
\[
[B_\alpha(r)(f_b-1)]_+\le \frac{2MK}{m}.
\]
Similarly, for $e\le e_{\max}\le K/m\le1/2$, \eqref{eq:dt-static-bounds} gives
\[
[cY(f_e-1)]_+\le\frac{2MK}{m},
\qquad
[cY(1-f_e)]_+\le\frac{2M}{m}.
\]
Concavity of $\phi$, together with \eqref{eq:dt-linear-identities}, now proves \eqref{eq:dt-upper-win}--\eqref{eq:dt-upper-lose}.

Suppose next that $|\phi''|\le L$. Since $b,e\le K/m\le1/2$,
\[
d_b\le\frac{2K}{m},
\qquad
d_e\le\frac{2K}{m}.
\]
Taylor's theorem and \eqref{eq:dt-linear-identities} therefore give a curvature loss of at most $2LK^2/m$ on either branch. Combining this with the preceding coefficient bounds proves \eqref{eq:dt-regular-win} and \eqref{eq:dt-regular-lose}. Moreover,
\begin{equation}\label{eq:emax-first-order}
\phi'(\alpha)\alpha e_{\max}=\frac{\phi(\alpha)}m,
\end{equation}
because $B_\alpha(1)=((1-\alpha)/\alpha)\phi(\alpha)$. If $e_{\max}\le1-\alpha$, then $d_{e_{\max}}\le2K/m$, and Taylor's theorem, \eqref{eq:emax-first-order}, and $\phi(\alpha)\le M$ yield
\[
(m-1)\phi(\alpha^-(e_{\max}))
\ge (m-1)\left(\phi(\alpha)+\frac{\phi(\alpha)}m-\frac{2LK^2}{m^2}\right)
\ge m\phi(\alpha)-\frac{M+2LK^2}{m},
\]
which implies \eqref{eq:dt-regular-above}.

Finally assume $m\ge16$ and $\alpha\ge K/\sqrt m$. Since $\alpha\le1$, this also implies $m\ge K^2$. On a winning branch,
\[
d_b\le\frac{2K}{m}\le\frac{\alpha}{2},
\]
so the interpolation interval stays above $\alpha/2$. Lemma \ref{lem:bounded-analytic} bounds the curvature contribution by
\[
\frac{m-1}{2}\frac{8M}{\alpha}d_b^2
\le\frac{16MK^2}{m\alpha}
\le\frac{16MK}{\sqrt m}.
\]
Together with the coefficient error $2MK/m$, this proves \eqref{eq:dt-bounded-win}. On a losing branch the interpolation interval stays above $\alpha$, and $d_e\le2\alpha K/m$, so the curvature contribution is at most
\[
\frac{m-1}{2}\frac{4M}{\alpha}d_e^2
\le\frac{8M\alpha K^2}{m}
\le\frac{8MK}{\sqrt m}.
\]
The coefficient error is at most $2M/m$, proving \eqref{eq:dt-bounded-lose}. The same calculation with $e=e_{\max}$, together with \eqref{eq:emax-first-order}, loses at most $M/m+8MK/\sqrt m$, which is bounded by $20MK/\sqrt m$. This proves \eqref{eq:dt-bounded-above}.
\end{proof}

\begin{remark}[Measurability of the explicit policies]\label{rem:policy-measurability}
For bounded $Q$, the map $(\alpha,r)\mapsto B_\alpha(r)$ is Borel on $(0,1)\times[0,1]$ by its integral formula, and $\alpha\mapsto\mu_\alpha$ is continuous on $(0,1)$ by Lemma \ref{lem:bounded-analytic}. \Adversary{}'s draw can be implemented as
\[
S_\alpha=U^{\alpha/(1-\alpha)},
\qquad U\sim U[0,1].
\]
Thus the state-dependent rules used below define Borel rank-bid rules and probability kernels of the kind required in Lemma \ref{lem:barrier-propagation}. Together with the explicit endpoint and trivial-region actions, their recursive concatenations are well-defined Markov behavioral policies.
\end{remark}

\section{Finite-horizon bounds for bounded distributions}\label{sec:bounded-bounds}

\subsection{Logarithmic lower bound for tame distributions}\label{subsec:lower-log}

Let $H_0=0$ and $H_m=\sum_{k=1}^m k^{-1}$ for $m\ge1$.

\begin{theorem}[Logarithmic lower bound]\label{thm:lower-log}
Under Assumption \ref{ass:tame}, let $K=M/\gamma$. Then for all $T\ge1$ and $\alpha\in[0,1]$,
\[
\lowC_T(\alpha,D)
\ge T\phi(\alpha,D)-A_{\rm reg}H_T,
\qquad
A_{\rm reg}:=2MK+2LK^2.
\]
In particular, the error is $O_D(\log T)$ uniformly in $\alpha$.
\end{theorem}

\begin{proof}
Put $R_m=A_{\rm reg}H_m$, define the raw barrier
\[
\ell_m(\alpha)=m\phi(\alpha)-R_m,
\]
and use its nonnegative clipping
\[
L_m(\alpha)=\max\{0,\ell_m(\alpha)\}.
\]
We verify the lower condition of Lemma \ref{lem:barrier-propagation}. At $\alpha=0$, \Bidder{} bids $0$ and obtains $0=L_m(0)$. At $\alpha=1$, she bids one half of the current total budget in every round. She wins, the normalized share stays $1$, and
\[
\E[V]+L_{m-1}(1)
\ge \max\{0,m\E[V]-R_m\}=L_m(1).
\]

For an interior state, define the first-round rule by
\begin{equation}\label{eq:regular-bidder-policy}
\beta_m^{\rm reg}(\alpha,r)=
\begin{cases}
\displaystyle \frac{B_\alpha(r)}{m\mu_\alpha},
& K/m\le1/2\ \text{and}\ \alpha\ge K/m,\\[2mm]
0,&\text{otherwise}.
\end{cases}
\end{equation}
If $K/m>1/2$, then $m<2K$ and
\[
m\phi(\alpha)\le mM<2MK\le A_{\rm reg}\le R_m.
\]
If instead $\alpha<K/m$, then
\[
m\phi(\alpha)\le mM\alpha<MK\le A_{\rm reg}\le R_m.
\]
Thus in every inactive state $L_m(\alpha)=0$; bidding $0$ gives nonnegative stage utility and a nonnegative continuation barrier, so the lower one-step inequality holds.

Consider now an active state. The bid in \eqref{eq:regular-bidder-policy} is feasible. Fix a feasible \Adversary{} bid $e$. If $e\le e_{\max}$, set $Y=m\mu_\alpha e$ and
\[
t=\sup\{s\in[0,1]:B_\alpha(s)\le Y\},
\]
with $t=0$ below the support. Since $L_{m-1}\ge\ell_{m-1}$, the regular comparisons \eqref{eq:dt-regular-win}--\eqref{eq:dt-regular-lose} give
\begin{align*}
\mathcal J_{L_{m-1}}(\alpha;\beta_m^{\rm reg}(\alpha,\cdot),e)
\ge{}&(m-1)\phi(\alpha)-R_{m-1}\\
&+\int_t^1\bigl(Q(r)-B_\alpha(r)\bigr)\,dr
+\frac{\alpha}{1-\alpha}tY
-\frac{A_{\rm reg}}m.
\end{align*}
The static expression is at least $\phi(\alpha)$ by Lemma \ref{lem:static-saddle}. Since $R_m-R_{m-1}=A_{\rm reg}/m$, the right-hand side is $\ell_m(\alpha)$.

If $e>e_{\max}$, this can occur only when $e_{\max}<1-\alpha$. \Bidder{} surely loses the current round; monotonicity of $L_{m-1}$ is not needed, because $L_{m-1}\ge\ell_{m-1}$ and monotonicity of $\phi$ together with \eqref{eq:dt-regular-above} gives
\[
\mathcal J_{L_{m-1}}(\alpha;\beta_m^{\rm reg}(\alpha,\cdot),e)
\ge m\phi(\alpha)-R_{m-1}-\frac{A_{\rm reg}}m
=\ell_m(\alpha).
\]
In either branch the one-step payoff is also nonnegative, because stage utilities and $L_{m-1}$ are nonnegative. Hence it is at least $L_m(\alpha)=\max\{0,\ell_m(\alpha)\}$. Lemma \ref{lem:barrier-propagation} yields an explicit Markov policy guaranteeing $L_T(\alpha)$, and therefore
\[
\lowC_T(\alpha,D)\ge L_T(\alpha)\ge T\phi(\alpha)-A_{\rm reg}H_T.
\]
\end{proof}

\subsection{Logarithmic auxiliary upper bound}\label{subsec:upper-log}

The upper bound identifies the auxiliary value asymptotically.

\begin{theorem}[Logarithmic upper bound for bounded distributions]\label{thm:upper-log}
Let $D$ be any bounded nonnegative valuation distribution. If $D\not\equiv0$, put $M=\|Q_D\|_\infty$, $\gamma=\gamma_D$, and $K=M/\gamma$. Then for all $T\ge1$ and $\alpha\in[0,1]$,
\[
\upC_T(\alpha,D)
\le T\phi(\alpha,D)+A_{\rm up}H_T,
\qquad
A_{\rm up}:=2MK.
\]
For $D\equiv0$ the same statement holds with $A_{\rm up}=0$. In particular, the error is $O_D(\log T)$ uniformly in $\alpha$.
\end{theorem}

\begin{proof}
The zero distribution is immediate, so assume $D\not\equiv0$. Put $R_m=A_{\rm up}H_m$ and define
\[
u_m(\alpha)=m\phi(\alpha)+R_m,
\qquad
U_m(\alpha)=\min\{m\E[V],u_m(\alpha)\}.
\]
We verify the upper condition of Lemma \ref{lem:barrier-propagation}. At $\alpha=0$, \Adversary{} bids $0$. The state remains $0$, and
\[
U_{m-1}(0)=\min\{(m-1)\E[V],R_{m-1}\}
\le \min\{m\E[V],R_m\}=U_m(0),
\]
so the endpoint one-step inequality holds (the actual total payoff is in fact $0$). At $\alpha=1$, the universal payoff bound is $m\E[V]=U_m(1)$.

At an interior state, \Adversary{} uses the kernel
\begin{equation}\label{eq:upper-opponent-policy}
\lambda_m^{\rm up}(\alpha)=
\begin{cases}
\mathcal L\!\left(B_\alpha(S_\alpha)/(m\mu_\alpha)\right),
& K/m\le1/2\ \text{and}\ 1-\alpha\ge K/m,\\
\delta_0,&\text{otherwise},
\end{cases}
\end{equation}
where $\Pr(S_\alpha\le s)=s^{(1-\alpha)/\alpha}$.

First consider an inactive state. If $K/m>1/2$, then $mM<2MK\le A_{\rm up}\le R_m$, and hence $u_m(\alpha)\ge m\E[V]$. If $1-\alpha<K/m$, then
\[
m\bigl(\phi(1)-\phi(\alpha)\bigr)
\le mM(1-\alpha)<MK\le R_m,
\]
so again $u_m(\alpha)\ge m\E[V]$. Thus $U_m(\alpha)=m\E[V]$. Against the bid $0$, or indeed against any first-round \Adversary{} action, stage utility is at most $Q(r)$ and $U_{m-1}\le(m-1)\E[V]$; therefore every \Bidder{} rule has one-step payoff at most $m\E[V]=U_m(\alpha)$.

Now consider an active state. The bid in \eqref{eq:upper-opponent-policy} is feasible because it is at most $e_{\max}\le K/m\le1-\alpha$. Since $U_{m-1}\le u_{m-1}$, the upper comparisons \eqref{eq:dt-upper-win}--\eqref{eq:dt-upper-lose} and Lemma \ref{lem:static-saddle} give, for every \Bidder{} rank-bid rule,
\[
\mathcal J_{U_{m-1}}(\alpha;\beta,\lambda_m^{\rm up}(\alpha))
\le m\phi(\alpha)+R_{m-1}+\frac{A_{\rm up}}m
=u_m(\alpha).
\]
The same one-step payoff is at most $m\E[V]$ by the universal stage-payoff bound and $U_{m-1}\le(m-1)\E[V]$. It is therefore at most $U_m(\alpha)$. Lemma \ref{lem:barrier-propagation} yields the explicit recursive \Adversary{} policy and
\[
\upC_T(\alpha,D)\le U_T(\alpha)\le T\phi(\alpha)+A_{\rm up}H_T.
\]
\end{proof}

\begin{corollary}[Auxiliary bounds in the tame case]\label{cor:value}
Under Assumption \ref{ass:tame},
\[
T\phi(\alpha,D)-O_D(\log T)
\le
\lowC_T(\alpha,D)
\le
\upC_T(\alpha,D)
\le
T\phi(\alpha,D)+O_D(\log T),
\]
uniformly in $\alpha\in[0,1]$. If the lower and upper values coincide, and the common value is denoted by $C_T(\alpha,D)$, then
\[
C_T(\alpha,D)=T\phi(\alpha,D)+O_D(\log T).
\]
\end{corollary}

\begin{corollary}[Fair-floor guarantee in the tame case]\label{cor:regular-guarantee}
If $D_i$ satisfies Assumption \ref{ass:tame}, write
\[
K_i=\frac{M_i}{\gamma_i},
\qquad
A_{{\rm reg},i}=2M_iK_i+2L_iK_i^2.
\]
Then the explicit auxiliary policy \eqref{eq:regular-bidder-policy}, scaled by the current total virtual budget, transfers to the original repeated first-price game and guarantees
\[
\E[U_i^T]\ge T\phi(\alpha_i,D_i)-A_{{\rm reg},i}H_T.
\]
\end{corollary}

\begin{proof}
Use the strategy-transfer part of Lemma \ref{lem:aggregation} and Theorem \ref{thm:lower-log}.
\end{proof}

\subsection{Square-root lower bound for arbitrary bounded distributions}\label{sec:bounded-sqrt}

The logarithmic lower estimate used the uniform $C^2$ bound in Assumption \ref{ass:tame}. The logarithmic upper estimate did not. For arbitrary bounded nonnegative distributions, the same barrier-propagation argument still gives an $O(\sqrt T)$ lower error. The following theorem gives that result.

Throughout this section, $D$ is bounded, nonnegative, and not identically zero. Let $Q=Q_D$ and let
\[
M=\|Q\|_\infty<\infty.
\]
The case $D\equiv0$ is trivial.

\begin{theorem}[Square-root lower bound for bounded distributions]\label{thm:bounded-sqrt}
Let $D$ be any bounded nonnegative valuation distribution. If $D\not\equiv0$, put $M=\|Q_D\|_\infty$, $\gamma=\gamma_D$, and $K=M/\gamma$. Then for all $T\ge1$ and $\alpha\in[0,1]$,
\[
\lowC_T(\alpha,D)
\ge T\phi(\alpha,D)-A_{\rm low}\sqrt T,
\qquad
A_{\rm low}:=40MK=\frac{40M^2}{\gamma_D}.
\]
For $D\equiv0$ the statement is trivial.
\end{theorem}

\begin{proof}
Assume $D\not\equiv0$, put $A=A_{\rm low}$, and define
\[
\ell_m(\alpha)=m\phi(\alpha)-A\sqrt m,
\qquad
L_m(\alpha)=\max\{0,\ell_m(\alpha)\},
\]
with $L_0=0$. As in Theorem \ref{thm:lower-log}, the endpoint actions satisfy the lower barrier inequalities.

For an interior state, define
\begin{equation}\label{eq:sqrt-blue-policy}
\beta_m^{\rm sqrt}(\alpha,r)=
\begin{cases}
\displaystyle \frac{B_\alpha(r)}{m\mu_\alpha},
& m\ge16\ \text{and}\ \alpha\ge K/\sqrt m,\\[2mm]
0,&\text{otherwise}.
\end{cases}
\end{equation}
If $m<16$, then
\[
m\phi(\alpha)\le mM\le4M\sqrt m\le A\sqrt m.
\]
If $\alpha<K/\sqrt m$, then
\[
m\phi(\alpha)\le mM\alpha<MK\sqrt m\le A\sqrt m.
\]
Thus $L_m(\alpha)=0$ in every inactive state, and the zero bid verifies the one-step lower inequality.

In an active state, $m\ge K^2$ and $K/m\le1/\sqrt m\le1/4$, so the canonical bid is feasible and Lemma \ref{lem:dynamic-tangent}\textup{(iii)} applies. If $e\le e_{\max}$, the same generalized-cutoff calculation as in Theorem \ref{thm:lower-log}, now using $L_{m-1}\ge\ell_{m-1}$, gives
\[
\mathcal J_{L_{m-1}}(\alpha;\beta_m^{\rm sqrt}(\alpha,\cdot),e)
\ge m\phi(\alpha)-A\sqrt{m-1}-\frac{20MK}{\sqrt m}.
\]
If $e>e_{\max}$, monotonicity of $\phi$ and \eqref{eq:dt-bounded-above} give the same bound. Finally,
\[
A(\sqrt m-\sqrt{m-1})
=\frac{A}{\sqrt m+\sqrt{m-1}}
\ge\frac{A}{2\sqrt m}
=\frac{20MK}{\sqrt m},
\]
so every \Adversary{} bid gives at least $\ell_m(\alpha)$. The one-step payoff is nonnegative as well, hence at least $L_m(\alpha)$. Lemma \ref{lem:barrier-propagation} now gives an explicit Markov policy guaranteeing $L_T(\alpha)$, and therefore the stated raw lower bound.
\end{proof}

\begin{corollary}[Auxiliary bounds for bounded distributions]\label{cor:bounded-auxiliary}
If $D$ is bounded and nonnegative, then
\[
T\phi(\alpha,D)-O_D(\sqrt T)
\le
\lowC_T(\alpha,D)
\le
\upC_T(\alpha,D)
\le
T\phi(\alpha,D)+O_D(\log T)
\]
uniformly in $\alpha\in[0,1]$.
\end{corollary}

\begin{proof}
Combine Theorems \ref{thm:bounded-sqrt} and \ref{thm:upper-log} with the elementary weak-duality inequality $\lowC_T\le\upC_T$.
\end{proof}

\begin{corollary}[Fair-floor guarantee for bounded distributions]\label{cor:bounded-guarantee}
If $D_i$ is bounded, nonnegative, and not identically zero, put
\[
M_i=\|Q_{D_i}\|_\infty,
\qquad
K_i=\frac{M_i}{\gamma_{D_i}},
\qquad
A_{{\rm low},i}=40M_iK_i.
\]
Then the explicit auxiliary policy \eqref{eq:sqrt-blue-policy}, scaled by the current total virtual budget, transfers to the original repeated first-price game and guarantees
\[
\E[U_i^T]\ge T\phi(\alpha_i,D_i)-A_{{\rm low},i}\sqrt T.
\]
For $D_i\equiv0$, the zero strategy gives the statement with zero error.
\end{corollary}

\begin{proof}
Use the strategy-transfer part of Lemma \ref{lem:aggregation} and Theorem \ref{thm:bounded-sqrt}.
\end{proof}

\subsection{Finite-support distributions}\label{sec:finite}

For finite-support distributions, $Q$ and $B_\alpha$ may have flat intervals. Lemma \ref{lem:static-saddle} handles them by defining \Adversary{}'s static bid as the push-forward $Y=B_\alpha(S_\alpha)$. The generalized-cutoff proof treats the resulting atoms and ties directly.

\begin{corollary}[Finite-support distributions]\label{cor:finite-log}
If $D$ has finite nonnegative support and is not identically zero, then $D$ satisfies Assumption \ref{ass:tame}. Hence
\[
T\phi(\alpha,D)-O_D(\log T)
\le
\lowC_T(\alpha,D)
\le
\upC_T(\alpha,D)
\le
T\phi(\alpha,D)+O_D(\log T)
\]
uniformly in $\alpha\in[0,1]$.
\end{corollary}

\begin{proof}
The quantile function is a bounded nondecreasing step function. Therefore $\phi$ is a finite linear combination of terms
\[
\alpha\bigl(s_2^{1/\alpha}-s_1^{1/\alpha}\bigr),
\qquad 0\le s_1<s_2\le1,
\]
with nonnegative coefficients. Each term extends to a $C^2$ function on $[0,1]$: if $s\in(0,1)$, then $\alpha s^{1/\alpha}=\alpha e^{(\log s)/\alpha}$ and all singular-looking derivatives vanish at $0$; the endpoint $s=1$ contributes only the linear term $\alpha$. Thus $\phi$ is $C^2$ with bounded second derivative. Since $D$ is not identically zero, Lemma \ref{lem:bounded-analytic} gives $\gamma_D=\phi'(1)>0$. Assumption \ref{ass:tame} follows, and Corollary \ref{cor:value} gives the estimate.
\end{proof}

\section{Unbounded distributions: truncation and moment rates}\label{sec:truncation}

Truncation extends the asymptotic result to all integrable nonnegative distributions. The lower estimate uses a truncated-value \Bidder{} strategy; the upper estimate uses a truncated-value \Adversary{} strategy and charges the tail as an additive loss.

\begin{corollary}[Quantitative truncation and integrable asymptotic bounds]\label{cor:integrable-guarantee}
Let $D$ be any integrable nonnegative valuation distribution. For $M\in[0,\infty)$, let $D^{(M)}$ be the law of $V\wedge M$ and write
\[
\tau(M)=\E[(V-M)_+].
\]
Let $A_M^{\rm low}$ and $A_M^{\rm up}$ be the constants furnished by Theorems \ref{thm:bounded-sqrt} and \ref{thm:upper-log}, respectively, when applied to $D^{(M)}$. Then, uniformly in $\alpha\in[0,1]$,
\begin{align}
\lowC_T(\alpha,D)&\ge T\phi(\alpha,D)-T\tau(M)-A_M^{\rm low}\sqrt T,\label{eq:quant-trunc-lower}\\
\upC_T(\alpha,D)&\le T\phi(\alpha,D)+T\tau(M)+A_M^{\rm up}H_T.\label{eq:quant-trunc-upper}
\end{align}
In particular,
\[
\lowC_T(\alpha,D)\ge T\phi(\alpha,D)-o_D(T),
\qquad
\upC_T(\alpha,D)\le T\phi(\alpha,D)+o_D(T).
\]
Moreover, for every fixed truncation level $M$, the explicit bounded-law policy for $D_i^{(M)}$, applied to $V_i\wedge M$ and scaled by the current total virtual budget, transfers to the original $n$-player game and gives player $i$ the quantitative guarantee
\[
\E[U_i^T]
\ge T\phi(\alpha_i,D_i)-T\tau_i(M)-A_{i,M}^{\rm low}\sqrt T,
\]
where $\tau_i(M)=\E[(V_i-M)_+]$ and $A_{i,M}^{\rm low}$ is the bounded-law constant for $D_i^{(M)}$.
Consequently, if every $D_i$ is integrable and nonnegative, a diagonal choice of truncation level gives each player the fair-floor guarantee
\[
\E[U_i^T]\ge T\phi(\alpha_i,D_i)-o_{D_i}(T).
\]
\end{corollary}

\begin{proof}
For $M\in[0,\infty)$, let $D^{(M)}$ be the law of $W=V\wedge M$, and write $\phi_M(\alpha)=\phi(\alpha,D^{(M)})$. A \Bidder{} strategy for $D^{(M)}$ can be used in the original game by replacing the observed value $V$ with $W$ when computing the bid. Whenever the strategy wins, the actual payoff $V$ is at least the truncated payoff $W$. Hence
\[
\lowC_T(\alpha,D)\ge \lowC_T(\alpha,D^{(M)}).
\]
By the choice of $A_M^{\rm low}$,
\[
\lowC_T(\alpha,D^{(M)})\ge T\phi_M(\alpha)-A_M^{\rm low}\sqrt T
\]
uniformly in $\alpha$.

For the upper estimate, use in the original game the explicit recursive \Adversary{} policy furnished by Theorem \ref{thm:upper-log} for the truncated law $D^{(M)}$, with the same public history of bids and normalized shares. Fix any \Bidder{} strategy in the original game. Since
\[
V\one_{\{\Bidder{}\text{ wins}\}}
\le
W\one_{\{\Bidder{}\text{ wins}\}}+(V-M)_+,
\]
the total expected payoff is at most the expected truncated payoff plus $T\E[(V-M)_+]$.

A player in the truncated game can reproduce the informational advantage of observing $V$ as follows. Whenever $W<M$, set $\widetilde V=W$. Whenever $W=M$, privately draw $\widetilde V$ from the conditional law of $V$ given $V\ge M$; if this event has probability zero, no auxiliary draw is needed. Retain all such draws as part of the private history and feed the sequence $(\widetilde V_t)$ into the original \Bidder{} strategy. Conditional on the truncated observations, these auxiliary draws have exactly the conditional law of the original values, independently across rounds. Hence, against the fixed \Adversary{} policy, the induced bids and public states have the same joint law as in the original game, while the stage payoff being evaluated is $W$. Such auxiliary sampling is admissible private behavioral randomization in the truncated game. The explicit upper policy therefore bounds the expected truncated payoff by $T\phi_M(\alpha)+A_M^{\rm up}H_T$. Taking the supremum over original-game \Bidder{} strategies gives
\[
\upC_T(\alpha,D)
\le T\phi_M(\alpha)+A_M^{\rm up}H_T+T\E[(V-M)_+].
\]

Finally, for $\alpha\in(0,1]$,
\[
0\le \phi(\alpha,D)-\phi_M(\alpha)
\le \int_0^1 \bigl(Q_D(u)-Q_{D^{(M)}}(u)\bigr)\,du
=\E[(V-M)_+],
\]
because $u^{1/\alpha-1}\le1$; the same bound is trivial at $\alpha=0$. Combining these inequalities proves \eqref{eq:quant-trunc-lower} and \eqref{eq:quant-trunc-upper}. Since $D$ is integrable, the tail expectation tends to $0$ as $M\to\infty$.

Choose $M_k\uparrow\infty$ so that
\[
\varepsilon_k:=\E[(V-M_k)_+]+\frac1k\downarrow0
\]
after passing to a subsequence if necessary. Choose increasing thresholds $T_k$ so large that, for every $T\ge T_k$,
\[
\frac{A_{M_k}^{\rm low}}{\sqrt{T}}\le\varepsilon_k,
\qquad
\frac{A_{M_k}^{\rm up}H_T}{T}\le\varepsilon_k.
\]
For $T_k\le T<T_{k+1}$, use the truncation level $M_k$. The two preceding estimates give
\[
\frac{T\phi(\alpha,D)-\lowC_T(\alpha,D)}{T}\le 2\varepsilon_k,
\qquad
\frac{\upC_T(\alpha,D)-T\phi(\alpha,D)}{T}\le 2\varepsilon_k,
\]
uniformly in $\alpha$, and these bounds tend to $0$. The quantitative original-game statement and its diagonal $o(T)$ consequence follow by transferring the explicit truncated-value policies through Lemma \ref{lem:aggregation}.
\end{proof}

\begin{lemma}[Explicit quadratic dependence under truncation]\label{lem:trunc-constant-growth}
Let $D$ be a nonnegative distribution that is not identically zero, and let $D^{(M)}$ be the law of $V\wedge M$. There are $M_0\ge1$ and $\gamma_0>0$, depending only on $D$, such that for every $M\ge M_0$ the constants in Theorems \ref{thm:upper-log} and \ref{thm:bounded-sqrt} may be chosen as
\[
A_M^{\rm up}=2MK_M,
\qquad
A_M^{\rm low}=40MK_M,
\qquad
K_M:=\frac{M}{\gamma_{D^{(M)}}},
\]
and satisfy, with $C_D:=42/\gamma_0$,
\begin{equation}\label{eq:trunc-constant-explicit}
A_M^{\rm up}+A_M^{\rm low}
\le C_D M^2.
\end{equation}
\end{lemma}

\begin{proof}
Choose $M_0\ge1$ so that $D^{(M_0)}$ is not identically zero, and put $\gamma_0=\gamma_{D^{(M_0)}}>0$. Since $Q_{D^{(M)}}=Q_D\wedge M$ increases pointwise with $M$, formula \eqref{eq:bounded-phi-prime} implies
\[
\gamma_{D^{(M)}}\ge\gamma_0
\qquad(M\ge M_0).
\]
Also $\gamma_{D^{(M)}}\le M$, so $K_M\ge1$, while $K_M\le M/\gamma_0$. Applying the proofs of Theorems \ref{thm:upper-log} and \ref{thm:bounded-sqrt} with the bound $\|Q_{D^{(M)}}\|_\infty\le M$ gives the displayed choices. Therefore
\[
A_M^{\rm up}+A_M^{\rm low}
=2MK_M+40MK_M
=42MK_M
\le\frac{42}{\gamma_0}M^2.
\]
\end{proof}

\begin{corollary}[Finite-moment truncation rates]\label{cor:finite-moment-rate}
Let $D$ be nonnegative and suppose $\E[V^p]<\infty$ for some $p>1$. Put
\[
\beta_p=\frac{p+3}{2p+2}.
\]
Then, uniformly in $\alpha\in[0,1]$,
\begin{align}
\lowC_T(\alpha,D)
&\ge T\phi(\alpha,D)-O_D\bigl(T^{\beta_p}\bigr),
\label{eq:moment-lower-rate}\\
\upC_T(\alpha,D)
&\le T\phi(\alpha,D)
+O_D\left(
T^{2/(p+1)}H_T^{(p-1)/(p+1)}
\right).
\label{eq:moment-upper-rate}
\end{align}
Both errors are $o(T)$. Since $2/(p+1)<\beta_p$, the logarithmic factor in \eqref{eq:moment-upper-rate} is absorbed by $T^{\beta_p-2/(p+1)}$. In particular, the weaker symmetric display
\[
T\phi(\alpha,D)-O_D(T^{\beta_p})
\le\lowC_T(\alpha,D)
\le\upC_T(\alpha,D)
\le T\phi(\alpha,D)+O_D(T^{\beta_p})
\]
also holds. Consequently, if player $i$ has a finite $p_i$-th moment for some $p_i>1$, then she can guarantee
\[
\E[U_i^T]
\ge T\phi(\alpha_i,D_i)-O_{D_i}(T^{\beta_{p_i}}),
\qquad
\beta_{p_i}=\frac{p_i+3}{2p_i+2}.
\]
\end{corollary}

\begin{proof}
The case $D\equiv0$ is trivial. For $M\ge1$,
\[
\tau(M)=\E[(V-M)_+]\le \E[V^p]M^{1-p}.
\]
Use Corollary \ref{cor:integrable-guarantee} and Lemma \ref{lem:trunc-constant-growth}. For the lower bound, take
\[
M_{\rm low}=T^{1/(2p+2)},
\]
rounded upward and enlarged to be at least $M_0$. Then
\[
T\tau(M_{\rm low})+C_D M_{\rm low}^2\sqrt T
=O_D\bigl(T^{(p+3)/(2p+2)}\bigr),
\]
which proves \eqref{eq:moment-lower-rate}.

The upper estimate may use a different truncation level. Take
\[
M_{\rm up}=\left(\frac{T}{H_T}\right)^{1/(p+1)},
\]
again rounded upward and enlarged when necessary. Both terms in the upper truncation error then have the same order:
\begin{align*}
T\tau(M_{\rm up})
&=O_D\left(T^{2/(p+1)}H_T^{(p-1)/(p+1)}\right),\\
C_D M_{\rm up}^2H_T
&=O_D\left(T^{2/(p+1)}H_T^{(p-1)/(p+1)}\right).
\end{align*}
This proves \eqref{eq:moment-upper-rate}. The original-game guarantee follows by transferring the explicit lower truncated-value policy through Lemma \ref{lem:aggregation}.
\end{proof}

\section{Further questions and extensions}\label{sec:remaining}

The results leave several natural questions open.

\begin{enumerate}[label=\textbf{Q\arabic*.}, leftmargin=2.5em]
\item \textbf{Exact finite-horizon minimax value.} A dynamic minimax theorem incorporating simultaneous bids, public-state feedback, and the tie rule against \Bidder{} would identify a single exact value at every finite horizon.

\item \textbf{Logarithmic lower bounds beyond the tame class.} For arbitrary bounded distributions the upper estimate is logarithmic, but the lower estimate proved here is only $O(\sqrt T)$. The obstruction is the lower Taylor remainder near small $\alpha$, where only $|\phi''(\alpha)|\lesssim 1/\alpha$ is available in general. Improving this would require either a sharper one-sided expansion or a barrier adapted to the local curvature of $\phi$.

\item \textbf{Sharper rates for unbounded distributions.} Corollary \ref{cor:finite-moment-rate} gives separate optimized lower and upper rates under any finite $p>1$ moment. Logarithmic or square-root rates for genuinely unbounded laws would require stronger tail assumptions or a sharper-than-quadratic analysis of the truncation constants.

\item \textbf{Surplus above the floors and utilitarian welfare.} The auction allocates profile-specific slack through bids. A direct equilibrium analysis may yield total-welfare bounds beyond those obtained by summing the individual guarantees.

\item \textbf{Bounded or lossy transfers.} The virtual-budget auction studied here is a pure NTU mechanism. A genuine intermediate model arises when real transfers are allowed only up to a cap, or when a payment costs the payer more than the utility it creates for its recipients. It is natural to ask whether limited real transfers combined with virtual-budget bidding can implement the same guaranteed-utility target with smaller losses, greater welfare, or a simpler rule than the general construction.
\end{enumerate}

\section{Conclusion}

The repeated virtual first-price auction has asymptotic Pareto price of anarchy at most $1.283$. Every sequence of equilibria inherits this asymptotic bound because each player can deviate to a strategy that secures her fair-floor utility against arbitrary joint behavior of the others. The feasible-region estimate behind the constant and the cross-instance Pareto maximality of the fair-floor target are established in \citet{Csoka2026}; the contribution here is to implement that target through a fixed first-price credit rule.

The symmetric heuristic explains the mechanism. Monotone rank bidding points toward efficient allocation, support indifference points toward tightness of the guarantee, and exact first-price aggregation turns all opponents into a single auxiliary \Adversary{}. The static tangent saddle and finite-horizon barriers make this local picture rigorous. Under the direct second-price rule with more than two players, opponents can make one another spend more than the focal player's bid, so the exact aggregation disappears.

The proved lower loss is $O_D(\log T)$ for tame bounded laws, $O_D(\sqrt T)$ for all bounded laws, and $o_D(T)$ for every integrable law. Thus the price-of-anarchy statement rests on a stronger playerwise conclusion: the auction supplies robust security levels, not merely good welfare at a selected equilibrium.

\bibliographystyle{plainnat}
\bibliography{repeated_first_price}

\end{document}